\documentclass[12pt,notitlepage]{iopart}

\usepackage{iopams}

\usepackage{graphicx}
\usepackage{amsthm}

\newtheorem{theorem}{Theorem}[section]
\newtheorem{example}[theorem]{Example}

\newtheorem{corollary}[theorem]{Corollary}
\newtheorem{proposition}[theorem]{Proposition}

\newcommand{\tfrac}[2]{\textstyle{\frac{#1}{#2}}}

\newcommand{\varPi}{\Pi}

\newcommand{\R}{\mathbb R} 
\newcommand{\N}{\mathbb N} 

\newcommand{\sfa}{\mathsf{A}}

\newcommand{\sfc}{\mathsf{C}}

\newcommand{\sfe}{\mathsf{E}}
\newcommand{\sff}{\mathsf{F}}
\newcommand{\sfg}{\mathsf{G}}

\newcommand{\sfm}{\mathsf{M}}
\newcommand{\sfn}{\mathsf{N}}
\newcommand{\sfq}{\mathsf{Q}}
\newcommand{\sfp}{\mathsf{P}}
\newcommand{\p}{\mathsf{p}}

\newcommand{\hi}{\mathcal H} 
\newcommand{\ip}[2]{\langle {#1}|{#2}\rangle} 
\newcommand{\tra}[1]{\mathrm{tr}\left[ {#1} \right]} 
\newcommand{\no}[1]{||{#1}||} 
\newcommand{\kb}[2]{|#1\,\rangle\langle\,#2|} 
\newcommand{\fii}{\varphi}
\newcommand{\veps}{\varepsilon}
\newcommand{\eps}{\epsilon}

\newcommand{\qhat}{Q}
\newcommand{\phat}{P}

\newcommand{\bre}{\mathcal{B}(\R)} 
\newcommand{\brr}{\mathcal{B}(\R^2)} 
\newcommand{\prob}{\mathsf{p}} 

\newcommand{\cD}{\mathcal{D}}

\newcommand{\W}{{\mathcal W}}

\newcommand{\deq}[1]{\W_{\varepsilon_1}({#1},\sfq)}
\newcommand{\dep}[1]{\W_{\varepsilon_2}({#1},\sfp)}

\newcommand{\w}{{\mathcal W}^0}

\newcommand{\delqo}[1]{\w_{\varepsilon_1,\delta}({#1},\sfq)}
\newcommand{\delpo}[1]{\w_{\varepsilon_2,\delta}({#1},\sfp)}

\newcommand{\deleo}{\w_{\varepsilon,\delta}({\sfe_1},\sfe)}
\newcommand{\deeo}{\w_{\varepsilon}({\sfe_1},\sfe)}

\newcommand{\dele}[1]{\mathcal{W}_{\varepsilon,\delta}({#1},\sfe)}
\newcommand{\dee}[1]{\mathcal{W}_{\varepsilon}({#1},\sfe)}

\newcommand{\bede}[1]{\beta_{\veps,\delta}({#1},\sfe)}
\newcommand{\bee}[1]{\beta_{\veps}({#1},\sfe)}

\newcommand{\bom}{\tau}
\newcommand{\iqd}{J_{q;\delta}}
\newcommand{\ixd}{J_{x;\delta}}
\newcommand{\ipd}{J_{p;\delta}}
\newcommand{\iqw}{J_{q;w}}

\newcommand{\ipw}{J_{p;w'}}
\newcommand{\intv}[2]{J_{{#1};{#2}}}

\newcommand{\su}{{\rm supp}}

\newcommand{\h}[1]{\mathcal{#1}}

\newcommand{\hM}{{\h M}}

\newcommand{\epno}{\eps_{\footnotesize\textsc{no}}}

\def\supp{{\mathop{\rm supp}\nolimits}\,}

\def\sab{_{\alpha\mkern1mu\beta}}

\begin{document}

\title[Error and unsharpness in approximate joint measurements]
{Error and unsharpness in approximate joint measurements of position and momentum}

\author
{Paul Busch}
\address{Department of Mathematics, University of York, York, UK}
\ead{paul.busch@york.ac.uk}

\author
{David B.~Pearson}
\address{Department of Mathematics, University of Hull, Hull, UK}
\ead{d.b.pearson@hull.ac.uk}



\begin{abstract}
\noindent 
In recent years, novel quantifications of measurement error in quantum mechanics
have for the first time enabled precise formulations of Heisenberg's famous but often 
challenged measurement uncertainty relation. These relations take the form of a 
trade-off for the necessary errors in  joint approximate measurements of position and momentum and other incompatible pairs of observables. Here we review some of these error measures, examine their properties and suitability, and compare their relative strengths as criteria for ``good" approximations. \\
\end{abstract}

\maketitle



\section{Introduction}\label{intro}

In recent years, Heisenberg's uncertainty principle has received renewed
attention and scrutiny.  The principle is often loosely associated with three sets 
of ideas -- preparation uncertainty, joint measurement error trade-offs, and error-disturbance trade-offs.
While the first of these is uncontroversial, the latter two are subjects of an ongoing controversy. 

For many decades, the only formally and operationally well-defined form of uncertainty relation known 
in the physics literature was the familiar {\em preparation uncertainty relation} for 
standard deviations of, say, position and momentum,
\begin{equation}\label{eqn:state-ur}
\Delta(Q,\rho)\,\Delta(P,\rho)\ge \tfrac{1}{2}\hbar.
\end{equation}
This relation is a statement about the widths of the probability distributions $\rho^Q,\rho^P$ of 
the position $Q$ and momentum $P$ in a state $\rho$, and it can be 
tested by measuring position and momentum in separate runs of experiments
on particles prepared in the same state $\rho$. 

Notwithstanding this clear-cut interpretation, the relation (\ref{eqn:state-ur}) is often 
paraphrased as constituting a limitation of the accuracies of any attempted joint 
measurements of position and momentum. This unjustified conflation has equally often 
been criticised, but then it happened not seldom that the critics (or their readers) jumped 
to the conclusion that the uncertainty principle has nothing to do with the possibility or 
impossibility of joint measurements of position and momentum.

The joint measurement uncertainty question was brought into focus with these conflicting views 
but until recently a rigorous investigation of the problem has remained outstanding. The first seemingly plausible
attempt offered at quantifying measurement errors in quantum mechanics is based on the concept of
{\em noise operator} that was introduced into quantum optical amplifier theory in the 1960s and soon
after applied in the measurement context. For a brief history of the development of the notion of 
{\em noise-(operator) based error}, we refer the reader to \cite{BLW2013a}. On the basis of this 
{\em state-dependent} error measure it appeared that joint measurement error relations are much weaker than
suggested by the Heisenberg form (\ref{eqn:state-ur}); this has led to claims of a violation or circumvention
of Heisenberg-type measurement uncertainty relations, both theoretically (e.g., \cite{Ozawa03,Hall04})
and experimentally (e.g., \cite{Erh12,Roz12}). As shown in \cite{BLW2013a}, however, the noise-based error
measures do not purely quantify errors but also contain contributions of preparation uncertainty; moreover,
they are of limited operational significance as state-specific error measures. 

In the meantime, different measures of measurement error were introduced that quantify the performance of
measuring devices and as such  are state-independent. For these measures, joint measurement trade-off 
relations have been formulated and proven. This development, which is reviewed in \cite{BuHeLa07}, was
enabled by making full use of the operational possibilities of quantum mechanics, notably by the generalised 
representation of observables as positive operator valued measures 
(POVMs). 

A key concept for this solution to the joint-measurement problem is that of an
{\em approximate} measurement of a given observable (represented by a POVM) $\sfe$, which is any 
measurement whose associated POVM $\sff$ is close to $\sfe$ in a suitable operationally
relevant sense. This has made it possible to overcome the obstacle of the 
noncommutativity of $\qhat$ and $\phat$, which precludes any {\em sharp} 
joint measurement of these observables: there are (generally noncommuting) 
{\em unsharp} observables $\sfm_1$, $\sfm_2$ that are jointly measurable and 
still constitute reasonable approximations of $Q$ and $P$, respectively. 

Two observables $\sfm_1,\sfm_2$ on $\bre$ are said to be {\em jointly measurable} if there 
is a third, {\em joint observable} $\sfm$ on $\brr$ such that $\sfm_1,\sfm_2$ are the
Cartesian marginals of $\sfm$, $\sfm_1(X)=\sfm(X\times\R)$, $\sfm_2(Y)=\sfm(\R\times Y)$.

Three proposed measures of error for approximate measurements of position and momentum and their associated uncertainty relations were briefly reviewed in \cite{BuHeLa07}: 
these were referred to as {\em standard error}, {\em (Monge) metric error}, and 
{\em error bar width}. The first of these is what we called {\em noise-based error (measure)} above (due to the limitations of this concept it seems inappropriate to refer to it as ``standard'').  In the meantime, measurement uncertainty relations   have been proven  for a wider class of metric error measures \cite{BLW2013b,BLW2013c}; these are based on the so-called Wasserstein distance of order $\alpha$ between probability measures on a metric space; here $\alpha$ is a parameter whose values range from 1 to $\infty$. The Monge metric corresponds to the value $\alpha=1$, while $\alpha=2$ is found to provide a natural operational quantum generalisation of the notion of root-mean-square (rms) error. 

It is the purpose of the present paper to analyze further the concept of error bar width introduced in \cite{BuPe07} and to study  its connections with the metric error measures. Some aspects of the noise operator based error will also be considered to the extent that they are useful as estimates of the other measures. 
In addition to measurement errors,  we also review other quantities describing the intrinsic unsharpness of the 
approximators of $\qhat,\phat$.

The proofs of measurement uncertainty relations given in \cite{Werner04,BuHeLa07,BLW2013c}  make it evident that the  measurement error relations follow mathematically from related preparation uncertainty relations. Versions of these latter relations will be the starting points for the presentations of error and unsharpness measures to be given below.

We begin with a brief introduction of the requisite mathematical tools.

\section{Preliminaries}\label{sec:prelim}
 
Our considerations will be based on the usual description of a physical
system in a separable complex Hilbert space $\hi$, with states being represented
as positive operators $\rho$ of trace 1 (also called {\em density operators}).  The convex set of all states will denoted  
$S$. Pure states correspond to unit vectors $\fii\in\hi$ or rather the associated rank-1
projections $\kb{\fii}{\fii}\equiv P_\fii$. Observables are represented as positive operator measures (POVMs) 
$\sfe$ on a measurable 
space $(\Omega,\Sigma)$ that are normalised, i.e., $\sfe(\Omega)=I$. In this paper $(\Omega,\Sigma)$ 
will be one of the Borel spaces $(\R,\bre)$ or $(\R^2,\brr)$. An observable $\sfe$ is called
\emph{sharp} if it is projection valued; otherwise $\sfe$ is an \emph{unsharp}
observable. We write $\rho^\sfe$ 
for the probability measure induced by a state $\rho$ and an observable
$\sfe$ via the formula $\rho^\sfe(X):=\tra{\rho \sfe(X)}$, $X\in \Sigma$.
We use the notation $\sfe[x^k]$, $k\in\N$, for the $k^{\mathrm{th}}$ moment operators  
$\int x^k\sfe(dx)$ of an observable $\sfe$ on $\bre$. These operators are defined
on their natural domains \cite{lahti1998moment} 
$D(\sfe[x^k])$ of all $\fii\in\hi$ for which the function $x\mapsto x^k$ 
is integrable with respect to the complex measures $\ip{\psi}{\sfe(dx)\fii}$ for all  
$\psi\in\hi$; this contains the square-integrability domain
$\{\fii\in\hi\,:\, \int x^{2k}\ip{\fii}{\sfe(dx)\fii}<\infty\}$.
The moments
of a probability measure $\prob$ on $\R$ will be denoted $\prob[x^k]$, $k\in\N$.


All uncertainty relations to be studied here will be formulated for the case of  
a quantum particle in one spatial dimension, with Hilbert space $\hi=L^2(\R)$ 
and canonical position and momentum operators $\qhat,\phat$, defined 
via $(\qhat\psi)(x)=x\psi(x)$, $(\phat\psi)(x)=-i\hbar (d\psi/dx)(x)$ on the usual
maximal domains ensuring selfadjointness. Generalizations to more degrees of 
freedom are straightforward. The spectral measures of  $\qhat$ and $\phat$ will 
be denoted $\sfq$ and $\sfp$, respectively.

An important class of POVMs representing approximations of position and momentum
are given by {\em smeared} position and momentum observables $\sfq^\mu,\sfp^\nu$, 
defined as convolutions of $\sfq,\sfp$ with probability measures $\mu,\nu$ on $\bre$:
\begin{eqnarray}
\sfq^\mu(X)&=\sfq*\mu(X)=\int_\R\mu(X-q)\,\sfq(dq),\nonumber\\
\sfp^\nu(Y)&=\sfp*\nu(Y)=\int_\R\nu(Y-p)\,\sfp(dp).\nonumber
\end{eqnarray}
The integrals are defined in the weak operator topology. 

We will make use of the important class of covariant phase space observables which 
is defined as follows.
Let $W(q,p)=\exp(iqp/2\hbar)\exp(-iq\phat/\hbar)\exp(ip\qhat/\hbar )$ be the Weyl operators
comprising an irreducible unitary projective representation of
the translations on phase space $\R^2$.
An observable $\sfg$ on $\R^2$ is called a \emph{covariant phase space observable} if it satisfies the covariance condition
\begin{equation*}
W(q,p)\sfg(Z)W(q,p)^*=\sfg(Z-(q,p)),\quad Z\in\brr.\nonumber
\end{equation*}
This is satisfied by the following family of observables $\sfg=\sfg^\bom$ 
on $\R^2$, which are thus covariant phase space observables:
\begin{equation}\label{gen-cov-obs}
\mathcal{B}(\R^2)\ni Z\mapsto \sfg^\bom(Z)=   \frac 1{2\pi\hbar}\int_Z W(q,p)\bom W(q,p)^*dq\,dp;
\end{equation}
the integral is defined in the weak operator topology and the operator density is
generated  by an arbitrary fixed positive operator $\bom$ of trace 1
(for details of the proof of these properties, see, e.g., \cite{CRQM}).
Moreover, \emph{every} covariant phase space
observable is of the form (\ref{gen-cov-obs}) for some positive operator $\tau$ of trace 1. This fundamental fact
is implied by results of \cite{PSAQT}  and \cite{Werner1984}  and has 
been made explicit in \cite{CaDeTo03}
using the theory of induced representations and in
\cite{KiLaYl06a} using the theory of integration with respect to
operator measures.

The marginal observables of $\sfg^\bom$ are smeared position and momentum
observables $\sfq^{\mu_\bom},\sfp^{\nu_\bom}$, where
$\mu_\bom:=\bom_\varPi^\sfq$ and $\nu_\bom:=\bom_\varPi^\sfp$ are the probability distributions of $Q$ and $P$
in the state described by $\bom_\Pi^{\phantom{P}}$, that is,
\begin{equation*}
\sfg^\bom_1=\sfq*\mu_\bom,\quad \sfg^\bom_2=\sfp*\nu_\bom.
\end{equation*}
Here $\bom_\varPi^{\phantom{P}}=\varPi\bom\varPi^*$ is the operator obtained from $\bom$
under the action of the parity transformation $\varPi$
($\varPi\fii(x)=\fii(-x)$). 

There is a simple but fundamental characterization of all pairs of smeared position
and momentum observables admitting a joint measurement.
\begin{theorem}\label{thm:jm-smeared-qp}
A pair of smeared position and momentum observables $\sfq^\mu,\sfp^\nu$ are 
jointly measurable exactly when there exists a covariant phase space observable 
$\sfg^\bom$ of which they are marginals. In that case, $\mu=\mu_\bom$, $\nu=\nu_\bom$.
\end{theorem}
This result has been obtained in a long series of investigations
by various authors, culminating and summarised in \cite{CaHeTo05}.

\section{Uncertainty: $\boldmath\alpha$-deviation and overall width}\label{sec:uncertainty} 

We will make use of the following measures of the widths of a probability distribution
$\prob:\bre\to[0,1]$ on $\R$. 
The \emph{standard deviation} $\Delta(\prob)$ is given by 
\begin{equation*}
\Delta(\prob):=\left(\int \left(x-\int x\prob(dx)\right)^2\prob(dx)\right)^{1/2}=
\left(\prob[x^2]-\prob[x]^2\right)^{1/2}.
\end{equation*}
The standard deviation of an observable $\sfe$ on $\bre$ in a state $\rho$ is
$\Delta(\sfe,\rho):=\Delta(\rho^\sfe)$. For vector states $\fii$ we write 
$\Delta(\sfe,\fii):=\Delta(\prob^\sfe_\fii)$. 

The standard deviation is a special case of the so-called (Wasserstein) $\alpha$-deviation:
\begin{equation*}
\Delta_\alpha(\prob):=\inf_y\left(\int \left|x-y\right|^\alpha\prob(dx)\right)^{1/\alpha},\quad 1\le\alpha<\infty.
\end{equation*}

The uncertainty relation for the standard deviations of position and momentum has recently been generalised to $\alpha$-deviation \cite{BLW2013b}.

\begin{theorem}[Preparation Uncertainty]\label{thm:prepURpq}
Let $\sfq$ and $\sfp$ be canonically conjugate position and momentum observables, and $\rho$ a density operator. Then, for any $1\leq \alpha,\beta<\infty$,
\begin{equation*}
  \Delta_\alpha(\rho^\sfq)\Delta_\beta(\rho^\sfp)\geq c\sab\hbar,
\end{equation*}
The constant $c\sab$ is connected to the ground state energy $g\sab$ of the Hamiltonian $H\sab=|Q|^\alpha+|P|^\beta$ by the equation
\begin{equation*}
    c\sab = \alpha^{\frac 1\beta}\beta^{\frac 1\alpha}\left(\frac{g\sab}{\alpha+\beta}\right)^{\frac 1\alpha+\frac 1\beta}.
\end{equation*}
The lower bound is attained exactly when $\rho$ arises from the ground state of the operator $H\sab$ by phase space translation and dilatation.
For $\alpha=\beta=2$, $H\sab$ is twice the harmonic oscillator Hamiltonian with ground state energy $g_{22}=1$, and  $c_{22}=1/2$.
\end{theorem}

The \emph{overall width} of $\prob$ (at confidence level $1-\veps$) is defined for $\veps\in[0,1)$ as
\begin{equation*}
W_\veps(\prob):=\inf\{w>0\,|\,\exists x\in\R: \,\prob([x-\tfrac w2,x+\tfrac w2])\ge 1-\veps\}.
\end{equation*}
This quantity is finite for any $\veps>0$. For the overall width of the distribution of an observable
$\sfe$ on $\bre$ in a state $\rho$ we will write $\W_\veps(\sfe,\rho):=W_\veps(\rho^\sfe)$. This describes
the extent to which the quantity described by $\sfe$ can be approximately localised. As shown in
\cite{UH1985}, the overall width is generally a more stringent measure of the spread of a distribution 
than the standard deviation.

In analogy to the uncertainty relation (\ref{eqn:state-ur}) for standard deviations,
the overall widths of the position and momentum distributions in a state $\rho$ 
also satisfy a trade-off relation: for positive 
$\veps_1,\veps_2$   satisfying $\veps_1+\veps_2<1$ there is a constant $K(\veps_1,\veps_2)>0$ such that
\begin{equation}\label{eqn:ow-ur}
W_{\veps_1}(\sfq,\rho)\cdot W_{\veps_2}(\sfp,\rho)\ge 2\pi\hbar\,K(\veps_1,\veps_2).
\end{equation}
An uncertainty relations of this form was first  presented in a somewhat implicit way 
in the context of signal analysis by Landau and Pollak in 1961 \cite{LaPo61}. 
Its explicit form was given by Uffink in 1990 \cite{Uffink90}:
\begin{equation}\label{eqn:ow-U}
K(\veps_1,\veps_2)=\left(\sqrt{(1-\veps_1)(1-\veps_2)}-\sqrt{\veps_1\veps_2}\right)^2.
\end{equation}
A somewhat simpler (but weaker) bound was given in \cite{BuHeLa07} 
using  elementary arguments:
\begin{equation*}
\widetilde{K}(\veps_1,\veps_2)=\left(1-(\veps_1+\veps_2)\right)^2 \le K(\veps_1,\veps_2).
\end{equation*}
Note that the two expressions coincide on the `diagonal':
\begin{equation*}
K(\veps,\veps)=\widetilde{K}(\veps,\veps)=(1-2\veps)^2.
\end{equation*}

\section{Intrinsic unsharpness: resolution width}\label{sec:unsharpness}

For two noncommuting observables to be jointly measurable, it is necessary that they are unsharp.
One expects intuitively that the required degree of their unsharpness depends on the extent of
their noncommutativity. 

We will see that two unsharp observables which approximate
position and momentum, respectively, cannot have arbitrarily small degrees of unsharpness  if
they are to be jointly measurable. We will use the following measures as indicators  of the
unsharpness of an observable $\sfe$ on $\R$.

For an observable $\sfe$ with support $\su(\sfe)$ given by $\R$ or a closed interval, the \emph{resolution width} (at confidence level
$1-\veps$) is defined as \cite{CaHeTo06}:
\begin{equation*}
\gamma_\veps(\sfe):=\inf\{ w>0\,| \forall x\in \R\,\exists\rho\in S:  
\rho^\sfe([x-\tfrac w 2,x+\tfrac w 2])\ge 1-\veps  \}.\
\end{equation*}
For a sharp observable $\sfe$ on $\bre$ 
 the resolution width
is $\gamma_\veps(\sfe)=0$ for all $\veps\in(0,1)$. It is worth noting that vanishing resolution width does not
require the observable to be sharp: in fact, any observable whose nonzero effects have norm 1 
has zero resolution width; an example is given by the so-called canonical phase observable \cite{Heinonen-etal2003}.

For the resolution width of $\sfq^\mu,\sfp^\nu$ we obtain (see also \cite{CaHeTo06}):
\begin{equation*}
\gamma_{\veps_1}(\sfq^\mu)=W_{\veps_1}(\mu),\qquad
\gamma_{\veps_2}(\sfp^\nu)=W_{\veps_2}(\nu).
\end{equation*}

If a pair of observables $\sfq^\mu$, $\sfp^\nu$ is jointly measurable, 
their resolution widths are determined
by the probability measures $\mu=\mu_\bom,\nu=\nu_\bom$ which
obey the uncertainty relations (\ref{eqn:state-ur}) and (\ref{eqn:ow-ur}); we thus obtain:
\begin{equation*}
\gamma_{\veps_1}(\sfq^{\mu_\bom})\,\gamma_{\veps_2}(\sfp^{\nu_\bom})=
W_{\veps_1}(\sfq,\bom)\, W_{\veps_2}(\sfp,\bom)\ge 2\pi\hbar\,
K(\veps_1,\veps_2).\label{eqn:gam-gam-cov-ur}
\end{equation*}
The last inequality holds for any $\veps_1,\veps_2>0$ with $\veps_1+\veps_2<1$.

\section{Error measures I: Distance between observables}\label{sec:dist}

We review three distinct measures of error which quantify the difference
between an observable $\sfe$ on $\bre$ to be approximated and the approximator $\sff$,
which is also a POVM on $\bre$. Any error measure should be \emph{operationally significant} in the sense that it  quantifies the difference between the distributions $\rho^{\sff}$ and  $\rho^{\sfe}$. We begin with a family of metric error measures.

\subsection{Wasserstein $\alpha$-distance: definition.}

Next we briefly review a family of distances on the set of observables on $\R$ that was used in \cite{BLW2013b}
to formulate measurement uncertainty relations for canonically conjugate pairs of observables such as position and momentum. We adapt the presentation given there for general metric spaces to the case of $\R$.

For any two probability measures $\mu,\nu$ on $\R$ a {\it coupling} is defined to be a probability measure $\gamma$ on $\R\times \R$  with  $\mu$ and $\nu$ as the Cartesian marginals.
The set of couplings between $\mu$ and $\nu$ will be denoted $\Gamma(\mu,\nu)$. Then, for any $\alpha$, $1\leq \alpha<\infty$ the $\alpha$-distance (also Wasserstein $\alpha$-distance \cite{Villani}) of $\mu$ and $\nu$ is defined as
\begin{equation}\label{Wasserstein}
  \cD_\alpha(\mu,\nu)=  \inf_{\gamma\in\Gamma(\mu,\nu)}  \cD^\gamma_\alpha(\mu,\nu)=
  \inf_{\gamma\in\Gamma(\mu,\nu)}\left(\int |x-y|^\alpha\,d\gamma(x,y) \right)^{\frac1\alpha}
\end{equation}
For $\alpha=\infty$, one defines $\cD^\gamma_\infty(\mu,\nu)=\gamma-{\rm ess}\ \sup\{|x-y|\,|\, (x,y)\in\R\times\R\}$ and thus
\begin{equation}\label{Wassersteininfty}
   \cD_\infty(\mu,\nu) = \inf_{\gamma\in\Gamma(\mu,\nu)} \cD^\gamma_\infty(\mu,\nu).
\end{equation}
It turns out that $\cD^\gamma_\infty(\mu,\nu)$ actually depends only on the support of $\gamma$, that is, $\cD^\gamma_\infty(\mu,\nu)=\sup\{|x-y|\,|\, (x,y)\in\supp(\gamma)\}$.

The existence of an optimal coupling is known, for $1\leq\alpha<\infty$, see \cite[Theorem 4.1]{Villani}, the case $\alpha=\infty$ is shown in \cite[Theorem 2.6]{Jylha2014}, but it does not imply that  $D_\alpha(\mu,\nu)$  is finite.

When $\nu=\delta_y$ is a point measure, there is only one coupling between $\mu$ and $\nu$, namely the product measure $\gamma=\mu\times\delta_y$. In that case (\ref{Wasserstein}) and (\ref{Wassersteininfty}) describe the deviation of the measure $\mu$ from a point $y$. In particular,  the Wasserstein distances between point measures are seen to be extensions of the given metric for points, interpreted as point measures.
The metric can become infinite, but the triangle inequality still holds \cite[after Example~6.3]{Villani}. The proof relies on Minkowski's inequality and the use of a ``Gluing Lemma'' \cite{Villani}, which builds a coupling from $\mu$ to $\zeta$ out of couplings from $\mu$ to $\nu$ and from $\nu$ to $\zeta$. It also covers the case $\alpha=\infty$, which is not otherwise treated in \cite{Villani}.

We can now define the \emph{(Wasserstein) $\alpha$-distance between observables} $\sfe,\sff$ on $\R$:
\begin{equation*}
\Delta_\alpha(\sfe,\sff):=\sup_{\rho\in S}\cD_\alpha(\rho^\sfe,\rho^\sff).
\end{equation*}
These distances are operationally significant and  global error measures,
taking into account the largest possible deviations between corresponding probability measures of the 
observables being compared.

\subsection{Working with Wasserstein distances: Kantorovich duality.}

A powerful tool for working with the distance functions is a dual expression of the infimum over couplings as a supremum over certain other functions obtained by the Kantorovich duality.
In this context we exclude the case $\alpha=\infty$.

First we note that the ``gap inequality''
\begin{equation}\label{gap}
 \int\Phi(y)\,d\nu(y) \ -\  \int\Psi(x)\,d\mu(x) \leq \int |x-y|^\alpha\,d\gamma(x,y)
\end{equation}
holds for any pair of functions $(\Psi,\Phi)$ and any coupling $\gamma$ whenever the constraint
\begin{equation}\label{compete}
  \Phi(y)-\Psi(x)\leq |x-y|^\alpha
\end{equation}
is satisfied.
The Kantorovich Duality Theorem asserts that the gap is actually closed:
\begin{equation}\label{DKanto}
  \cD_\alpha(\mu,\nu)^\alpha=\sup_{\Phi,\Psi} \left\{ \int\Phi(y)\,d\nu(y) - \int\Psi(x)\,d\mu(x)\right\}
\end{equation}
where functions $\Phi$ and $\Psi$ satisfy (\ref{compete}).

When maximizing the left hand side of (\ref{gap}), one can naturally choose $\Phi$ as large as possible under the constraint (\ref{compete}), i.e.,
$\Phi(y)=\inf_x\{\Psi(x)+|x-y|^\alpha \}$, and similarly for $\Psi$. Hence one can choose just one variable $\Phi$ or $\Psi$ and determine the other by this
formula. In the case $\alpha=1$ the triangle inequality for the metric on $\R$ entails that one can take $\Phi=\Psi$. In this case (\ref{compete}) just asserts that
this function is Lipshitz continuous with respect to the metric on $\R$, with constant $1$. The left hand side of (\ref{gap}) is thus a difference of
expectation values of the given measures $\mu$, $\nu$.

It is of interest to note that the duality gap still closes if the set of functions $\Phi,\Psi$ is further restricted. 
The natural condition is, first of all, that $\Psi\in L^1(\mu)$. The statement of Kantorovich Duality in \cite[Thm.~5.10]{Villani} 
includes that the supremum (\ref{DKanto}) is attained also when one  restricts the set of functions to bounded continuous functions. 
In \cite{BLW2013b} it is shown that this set can be further restricted 
to positive continuous functions of compact support without changing the value of the supremum.

\subsection{Properties of the 1-distance.}

We now specialise to the case of the Wasserstein 1-distance ($\alpha=1$), also known as the  Monge metric. This was the choice of metric for the first formulation of a rigorous measurement uncertainty relation for position and momentum in \cite{Werner04}.

Denoting by $\Lambda$ the set of Lipshitz functions, that is, the bounded measurable functions $h:\R\to\R$ for which $|h(x)-h(y)|\leq |x-y|$, 
the Wasserstein-1 distance $\cD_1$ then becomes \cite{Werner04}
\[
\cD_1(\mu,\nu)=\sup_{h\in\Lambda}\,\left|\int hd\mu-\int hd\nu\right|.
\]
This gives rise to a metric on the set of observables on $\bre$ as follows.
We first recall that for any bounded measurable
function $h:\R\to\R$, the integral $\int_\R h\,d\sfe$ defines (in the weak
sense) a bounded selfadjoint operator, which we denote by $\sfe[h]$.
Thus, for any vector state $\fii$ the number $\ip{\fii}{\sfe[h]\fii}=\int_\R h\,d\ip\fii{\sfe(x)\fii}$ is well-defined.

The Wasserstein 1-distance between observables $\sfe$ and $\sff$ can then be expressed as
\begin{equation*}
\Delta_1(\sfe,\sff) := \sup_{\rho\in S}\, \sup_{h\in\Lambda}\,
\biggl| \tr\left[\rho\,( \sfe[h]-\sff[h]) \right]\biggr|=\sup_{h\in\Lambda}\,\biggl\|{\sfe[h]-\sff[h]}\biggr\|.
\end{equation*}

An observable $\sfc$ will be called a \emph{metric approximation} to 
$\sfa$ if $\Delta_1(\sfa,\sfc)<\infty$.

\begin{example}\label{ex:triv-dist}
Any trivial observables $\sfe=\mu\,I$ on $\R$ has infinite
distance from sharp position $\sfq$: $d(\mu\,I,\sfq)=\infty$.

Take the family of functions $h_n(x)=n-|x-c_n-n|$ if $|x-c_n-n|\le
n$, and $h_n(x)=0$ otherwise; here $(c_n)_{n\in\N}$ is an increasing
sequence of positive numbers still  to be determined. Note that
$h_n\in\Lambda$.
We have $||\sfq[h_n]||=h_n(c_n+n)=n$, so this approaches infinity as
$n\to\infty$.

For a  trivial observable $\sfe$ we get
$\sfe[h_n]=\int h_nd\mu\,I=:\mu(h_n)\,I$.
We show that for a suitable choice of the sequence $c_n$, one obtains
$\|\sfe[h_n]\|=\mu(h_n) \to 0$ as $n\to\infty$.

Let $c_n$ be such that the set $K_n=(-\infty,c_n]$ has measure
$\mu(K_n)> 1-1/n^2$, so that $\mu(\R\setminus
K_n)=\mu((c_n,\infty))<1/n^2$. Then
\begin{equation*}
\|\sfe[h_n]\|=\mu(h_n)=\int_{c_n}^{c_n+2n} h_n(x)\mu(dx)
\le n\mu((c_n,\infty))<1/n.
\end{equation*}
By the triangle inequality for norms we get
\[
\|\sfq(h_n)-\sfe(h_n)\|\ge\|\sfq(h_n)\|-\|\sfe(h_n)\|> n-1/n.
\]
It follows that the distance $\Delta_1(\sfe,\sfq)=\infty$. \qed

\end{example}

Our next example exhibits  functions of position $\sfq$ that may or may not
be good metric approximations to $\sfq$.
\begin{example}\label{ex:bdd-dist}
Let $g$ be a bounded measurable function on $\R$. Then the distance
of $\sfq$ and $\sfq\circ g^{-1}$ is infinite, $\Delta_1(\sfq,\sfq\circ g^{-1})=\infty$. For a function $f(x)=x+g(x)$ where
$g$ is bounded, the distance is finite: $\Delta_1(\sfq,\sfq\circ f^{-1})=\sup(|g|)$.

\emph{Proof.} We will show that $||\sfq[h_n]-\sfq\circ g^{-1}[h_n]||\to\infty$ as $n\to\infty$
for a suitable sequence of functions $h_n\in\Lambda$. To this end we
use the inequality 
$$
||\sfq[h_n]-\sfq\circ g^{-1}[h_n]||\ge \bigg| ||\sfq[h_n]|| -
||\sfq\circ g^{-1}[h_n]|| \bigg|,
$$ 
and  choose $h_n$ such that
$||\sfq[h_n]||\to\infty$ as $n\to\infty$,  while $||\sfq\circ g^{-1}[h_n]||$ will
remain bounded.

Let $|g(x)|\le g_0$.
Choose $h_n(x)=n-|x-n|$ if $|x-n|\le n$ and $h_n(x)=0$ otherwise.
Then we have $h_n\in\Lambda$. Further,
$\sfq[h_n]=\int h_n(x) \sfq(dx)$, so $\no{\sfq[h_n]}=n\to\infty$ as $n\to\infty$.
Next, we see that for $n > g_0$
\begin{equation*}
\sfq\circ g^{-1}[h_n]=\int h_n(t)\sfq\circ g^{-1}(dt)=\int h_n(g(x))\sfq(dx)
\end{equation*}
is a bounded operator since then $|h_n(g(x))|\le g_0$, and so
$\no{\sfq\circ g^{-1}[h_n]}\le g_0$.

To verify the second claim, we note that ${\sfq[h]-\sfq\circ f^{-1}[h]}=h(Q)-h(f(Q))$, and so for $h\in\Lambda$ and unit vector $\fii$,
\begin{eqnarray*}
\ip\fii{\bigl[h(Q)-h(f(Q))\bigr]^2\fii}&=&\int|\fii(x)|^2\left[h(x)-h(f(x))\right]^2dx\\
&\le&\int|\fii(x)|^2\left(x-f(x)\right)^2dx\\
&=&
\int|\fii(x)|^2g(x)^2dx\le\no{g(Q)}^2=(\sup|g|)^2,
\end{eqnarray*}
from which the claim follows. \qed

\end{example}

\begin{example}\label{ex:approx-pos-dist}
For smeared position and momentum observables $\sfq^\mu$,  $\sfp^\nu$,
the distances from $\sfq$ and $\sfp$ are
\begin{equation*}
\Delta_1(\sfq^\mu,\sfq)=\int |q|\,\mu(dq),\quad  \Delta_1(\sfp^\nu,\sfp)=\int |p|\,\nu(dp) .
\end{equation*}
(See \cite{Werner04}.) Thus, $\sfq^\mu$ and $\sfp^\nu$ are metric approximations
of $\sfq$ and $\sfp$
exactly  when these integrals are finite. 
\end{example}

\subsection{Measurement uncertainty relations for metric errors}\label{sec:distance}

The measurement uncertainty relation for the metric error associated with the 1-deviation (or Monge metric) 
proven in \cite{Werner04} has recently been generalised to all Wasserstein distances \cite{BLW2013b}. 
\begin{theorem}
\label{thm:main}
Let $M$ be a  phase space observable and $1\leq\alpha,\beta\leq\infty$. Then
\[
  \Delta_\alpha(\sfm_1,\sfq)\, \Delta_\beta(\sfm_2,\sfp)\               \geq\ c\sab\hbar
\]
provided that the quantities on the left hand side are finite. The constants $c\sab$ are the same as in 
Theorem~\ref{thm:prepURpq}.
\end{theorem}
The proof is analogous to that of Werner's original theorem for $1$-deviations: it proceeds by reduction 
to the covariant case, and the latter is immediately obtained by application of the preparation uncertainty 
relation of Theorem \ref{thm:prepURpq}.

The case where one of the distances is zero is in fact covered by the Theorem:  the other 
distance must then be infinite. In fact if (in the case of the 1-deviation) one has 
$\Delta_1(\sfm_1,\sfq)=0$, then $\sfm_1=\sfq$ and $\Delta_1(\sfm_2,\sfp)$ cannot be finite; otherwise
the associated covariant phase space observable would have to have $\sfq$ as its first marginal,
which is impossible. Hence Theorem \ref{thm:main} implies that whenever $\Delta_1(\sfm_1,\sfq)=0$
then $\Delta_1(\sfm_2,\sfp)=\infty$. It is an instructive exercise to verify this explicitly.

\begin{example}\label{ex3}
Let $\sfm$ be an observable on phase space $\R^2$ whose first marginal
$\sfm_1$ is  sharp position $\sfq$. Then the second
marginal $\sfm_2$ has infinite 1-distance from sharp momentum
$\sfp$.

\emph{Proof.} We note first that all positive operators (effects)
$\sfm_2(X)$ in the range of $\sfm_2$ commute with  $\qhat$ (see, e.g., \cite{CRQM}) and are thus
functions of $\qhat$. Thus one can write
$\sfm_2(X)=\int \sfq(dq)m(q,X)$, where the functions $m(\cdot,X)$ are
defined almost everywhere for all (Borel) subsets $X$ of $\R$, and
$X\mapsto m(q,X)$ is then a probability measure. We consider states
$\rho$ with the same fixed position distribution, $\rho^\sfq=\prob$,
and compute
\[
\tra{\rho \sfm_2(h)}=\int h(x)m_\prob(dx),\ m_\prob(X):=\int \prob(dq)m(q,X).
\]
We will let $h$ run through a family $h_n\in\Lambda$ and $\rho$ through a family
$\p_{\rho_n}\in S_q$ such that $\p_{\rho_n}^\sfq=\prob$ and  $\tra{\rho_n \sfm_2(h_n)}\to 0$, while $\tra{\rho_n\sfp(h_n)}\to\infty$.
This shows that $\Delta_1(\sfm_2,\sfp)=\infty$.

Choose $h_n$ as in Example \ref{ex:triv-dist}, where we have now $\mu=m_\prob$. This gives
$\tra{\rho_n \sfm_2(h_n)}\to 0$ for any $\rho_n$ (yet to be specified) with $\p_{\rho_n}^\sfq=\prob$.

Let $\rho_n=W(0,c_n+n-(c_1+1))\rho_1W(0,c_n+n-(c_1+1))^*$, with $\rho_1$ a state whose
momentum distribution is centered symmetrically at $c_1+1$, the peak location of $h_1$. Then
the momentum distribution of $\rho_n$ is centered at the peak location $c_n+n$ of $h_n$.
Also note that $\p_{\rho_n}^\sfq=\p_{\rho_1}^\sfq=:\prob$.
Specifically we take $\rho_n$ such that the densities $\p_{\rho_n}^\sfp(p)=\chi_{J_n}(p)$,
$J_n=[c_n+n-1/2,c_n+n+ 1/2]$. Then we have $\tra{\rho_n\sfp(h_n)}= n-1/4\to\infty$ as $n\to\infty$.
\qed
\end{example}

\section{Error measures II: error bar width}

\subsection{Gross error bar.}

We now present a definition of measurement error in terms of likely error intervals
that follows most closely the usual practice of calibrating measuring instruments. In the 
process of {\em calibration} of a measurement scheme, one seeks to obtain estimates of the
likely error and perhaps also the degree of disturbance that  the scheme contains. In order to 
estimate the error, one tests the device by applying it to a sufficiently large family of input 
states in which the observable one wishes to measure with this setup has fairly sharp values. 
The error is then characterised as an overall measure of the bias and 
the  width of the output distribution across a range of input values. Error bars give 
the minimal average interval lengths that one has to allow to contain all output values with a 
given confidence level.

For simplicity, we give the following definitions only for approximations of a sharp observable $\sfe$,
so that the assumption of localised input states $\rho$ can be described as $\rho^\sfe(J_{x;\delta})=1$, for
intervals $J_{x;\delta}:=[x-\delta/2,x+\delta/2]$, $x\in\R,\delta>0$.

Let $\sfe_1,\sfe$ be observables on $\R$ and $\sfe$ be sharp.
For each $\varepsilon\in(0,1)$, $\delta>0$, we define the \emph{error} of $\sfe_1$ relative to $\sfe$

\begin{eqnarray}
\dele{\sfe_1}:=&\inf\{w>0\,|\ \forall\ x\in\R\ \forall\rho\in S:\nonumber\\
&\qquad\qquad\qquad \rho^\sfe(J_{x;\delta})=1\Rightarrow \rho^{\sfe_1}(J_{x,w})\ge 1-\veps\}.\nonumber
\end{eqnarray}
The error describes the range within which the input values can be inferred 
from the output distributions, with confidence level $1-\varepsilon$, given initial localizations within 
$\delta$.
$\sfe_1$ is called an $\veps$-\emph{approximation} to $\sfe$ if
$\dele{\sfe_1}<\infty$ for all $\delta>0$. 
Note that the error is an increasing function of $\delta$, so that one can
define the \emph{(gross) error bar width} of $\sfe_1$ relative to $\sfe$:
\begin{equation*}
\dee{\sfe_1}:=\inf_\delta\dele{\sfe_1}=\lim_{\delta\to 0}\dele{\sfe_1}.
\end{equation*}
In the case $\dele{\sfe_1}=\infty$ for all $\delta>0$, we write $\dee{\sfe_1}=\infty$.

$\sfe_1$ will be called an \emph{approximation (in the sense of finite error bar width)} 
to $\sfe$ if $\dee{\sfe_1}<\infty$ for all $\veps\in(0,\frac 12)$.  The restriction to $\veps<\frac 12$
reflects the idea that a ``good" approximation should have confidence levels greater than $\frac 12$.

We note that if $\sfe_1$ is an approximation to $\sfe$, the map 
$\veps\mapsto \dele{\sfe_1}$ is a decreasing function of $\veps\in(0,\frac 12)$
for every $\delta>0$.

The following result shows that our definition is not empty.

\begin{proposition}\label{prop:smeared-pos-approx}
The smeared position and momentum observables $\sfq^\mu$, $\sfp^\nu$ are 
approximations (in the sense of finite error bar widths) to $\sfq$ and $\sfp$, respectively, for
any probability measures $\mu$, $\nu$.
\end{proposition}

\emph{Proof.} It is sufficient to consider the case of the position observable.
Let $\veps\in(0,1), \delta>0$ be given. We have to show that there is a finite number $w>0$ 
such that for all $q\in\R$ one has  $\rho^{\sfq^\mu}(\iqw)\ge1-\veps$ whenever
$\rho^\sfq(\iqd)=1$.

Let $q_0,w_0$ be such that $\mu(J_{q_0;w_0})\ge 1-\veps$. Then, if 
$w\ge 2|q_0|+w_0+\delta$, it follows that $\iqd\subseteq x+\iqw$ for all $x\in J_{q_0;w_0}$,
that is, $\rho^\sfq(x+\iqw)=1$ for all such $x$. Then:
\begin{eqnarray*}
\rho^{\sfq^\mu}(\iqw)&=\int\mu(dx)\rho^\sfq(x+\iqw)\ge
\int_{J_{q_0;w_0}}\mu(dx)\rho^\sfq(x+\iqw)\\
&=\mu(J_{q_0;w_0})\ge1-\veps.\qed
\end{eqnarray*}


\subsection{Properties of the error bar width}

It is not hard to construct approximations of $\sfq$ that do not share the translation covariance
of $\sfq$.

\begin{example}\label{ex:noncov-approx}
Let $f$ be a continuous function on $\R$ which is one-to-one and such that $f(q)-q$ is not constant but
$|f(q)-q|\le \alpha$ for
all $q\in \R$ and some fixed $\alpha>0$. An example is $f(q)=q+\frac 12\cos(q)$. Let $\sfq^\mu$ be a smeared position
observable. Then $\sfq^\mu\circ f^{-1}$ is a non-covariant approximation to $\sfq$ in
the sense of finite error bars.

\emph{Proof.}
Let $\veps\in(0,1)$and $\delta>0$ be given. We have to show that there is a finite positive $w$ such 
that for all $q\in\R$ and all $\rho$ with $\rho^\sfq(\iqd)=1$, then 
$\rho^{\sfq^\mu}(f^{-1}(\iqw))\ge1-\veps$.

We know that $\sfq^\mu$ is an approximation to $\sfq$. Hence there is $w'>0$ such that
for all $q\in\R$ and all $\rho$ with $\rho^\sfq(\iqd)=1$, we have 
$\rho^{\sfq^\mu}((\intv{q}{w'}))\ge1-\veps$.

Now take $w=w'+2\alpha$. This entails that $f^{-1}(\iqw)\supseteq\intv{q}{w'}$ for all $q\in\R$. Then, for
$q\in\R$ and $\rho$ such that $\rho^\sfq(\iqd)=1$ we obtain
\begin{equation*}
\rho^{\sfq^{\mu}\circ f^{-1}}(\iqw)=\rho^{\sfq^\mu}(f^{-1}(\iqw))\ge \rho^{\sfq^\mu}(\intv{q}{w'})\ge 1-\veps.
\end{equation*} 
Noting that $f^{-1}(\iqw+q')\ne f^{-1}(\iqw)+q' $ (since $f(q)-q$ is not constant) one concludes readily that 
$\sfq^{\mu}\circ f^{-1}$ is not covariant. \qed  
\end{example}

\begin{example}\label{ex:const-inf-err}
For any bounded Borel function $f$ on $\R$, the observable $\sfq\circ f^{-1}$  has infinite error 
bars with respect to $\sfq$.

\emph{Proof.}
Let $J$  be  a bounded interval which contains the range of $f$.  Then for any finite $w>0$, one can find
$q$ such that $\iqw\cap J=\emptyset$. Then $f^{-1}(\iqw)=\emptyset$ and so
$\rho^{\sfq\circ f^{-1}}(\iqw)=\rho^\sfq(f^{-1}(\iqw))=0$ for all $\rho$. 

It follows that $\W_{\varepsilon,\delta}(\sfq\circ f^{-1},\sfq)=\infty$ for all $\veps\in(0,1)$ and all $\delta>0$.
\qed
\end{example}

It is possible to characterise the case of an accurate
measurement of the sharp observable $\sfe$.
\begin{proposition}\label{prop:zero-err}
Let $\sfe_1$ be an approximation of the sharp observable $\sfe$. Then the following are equivalent:
\begin{enumerate}
\item[(a)] $\dele{\sfe_1}\le\delta$ for all $\veps\in(0,\frac12),\delta>0$;
\item[(b)] $\sfe_1=\sfe$.
\end{enumerate}
If either of these condition is fulfilled then $\dee{\sfe_1}=0$ for all $\veps\in(0,\frac12)$.
\end{proposition}

\emph{Proof.} Assume (b) holds. Let $\veps\in(0,\frac12)$, $\delta>0$.
Choose $w=\delta$; then for any $q\in\R$ and
any state $\rho$ with $\rho^\sfe(\ixd)=1$, we also have
$\rho^{\sfe_1}(\ixd)\ge 1-\veps$. This shows that $\dele{\sfe_1}\le\delta$.

Conversely, assume that (a) holds. Consider any
$\veps\in(0,\frac12),\delta>0$. For $w=\dele{\sfe_1}\le\delta$, we have, for
all $x\in\R$ and all $\rho$ with $\rho^\sfe(\ixd)=1$, that
$\rho^{\sfe_1}(\ixd)\ge\rho^{\sfe_1}(\intv{x}{w})\ge 1-\veps$. This entails for any vector state
$\fii$ for which $\sfq(\ixd)\fii=\fii$ that
$\ip{\fii}{\sfe_1(\ixd)\fii}\ge 1-\veps$. As this holds for any
$\veps\in(0,\frac12)$, it follows that $\ip{\fii}{\sfe_1(\ixd)\fii}=1$. This
entails that $\sfe(\ixd)\le \sfe_1(\ixd)$. Since $x\in\R$ and
$\delta>0$ are arbitrary, this operator inequality holds for any closed
interval $J=[a,b]$.

We show that then also $\sfe((a,b))\le E_1((a,b))$ for any open
interval.  Let $J_n$ be an increasing sequence of closed sets which
converges to a given open interval $(a,b)$. Put
$D_n:=\sfe_1(J_n)-\sfe(J_n)\ge O$. For any POVM $\sfn$ on $\bre$ we have
$\sfn(J_n)\to \sfn((a,b))$ (ultraweakly).
(This is a consequence of the regularity of Borel measures
on $\bre$, see, e.g.,  \cite{NOST}.) 
So we obtain $D_n\to \sfe_1((a,b))-\sfe((a,b))$ (ultraweakly),
and since  $D_n\ge O$, this limit operator is also nonnegative. In
this way we conclude that $\sfe(K)\le \sfe_1(K)$ for all open
intervals $K$. Similarly we can show that $\sfe((a,b])\le \sfe_1((a,b])$. 
Due to the normalization of both POVMs $\sfe,\sfe_1$, it follows that they must
coincide on all intervals and finally, since the intervals generate $\bre$, that they are identical. \qed

We remark that it is not known whether  the condition $\dee{\sfe_1}=0$ for all $\veps\in(0,\frac12)$
is sufficient to conclude that $\sfe_1=\sfe$.

\begin{proposition}\label{prop:errbar-res}
Let $\sfe_1,\sfe$ be observables with support $\R$, and $\sfe$ be a sharp observable. 
The error bar width of $\sfe_1$ relative to $\sfe$ is never smaller than the  intrinsic
resolution width of $\sfe_1$: 
\begin{equation*}
\dee{\sfe_1}\ge \gamma_\veps(\sfe_1).
\end{equation*}
\end{proposition}
The proof is given in \cite[Prop. 1]{BuPe07}.

\begin{corollary}
Let $\sfe$ be a sharp observable on $\bre$ with support $\R$. Any $\veps$-approximation $\sfe_1$ 
(supported on $\R$) of
$\sfe$ has finite resolution width, $\gamma_\veps(\sfe_1)<\infty$.
\end{corollary}


\subsection{Bias-free error and bias}

We show next how the gross error can be decomposed into a (positive) bias term and
a random error.
Let $\veps\in(0,1)$ and $\delta>0$ be given. Let $\sfe_1,\sfe$ be observables on $\R$ and $\sfe$ be sharp. 
Note that the condition $\rho^\sfe(\ixd)=1$ (for some $x\in\R$) can be expressed as 
$W_0(\rho^\sfe)\le\delta$.
We define the \emph{bias-free}, or {\em random error} $\deleo$ as follows:
\begin{equation*}
\deleo:=\sup\big\{W_\veps(\rho^{\sfe_1})\,|\, W_0(\rho^\sfe)\le\delta\big\}.
\end{equation*}
This is a measure of the overall minimal error, determined by the overall widths of all
output distributions, given input distributions supported in intervals $\iqd$.
If this quantity is finite for some $\delta_0$, it is an increasing function for all $\delta\le\delta_0$.
In that case we can define the \emph{bias-free error bar width},
\begin{equation*}
\deeo:=\lim_{\delta\to 0}\deleo.
\end{equation*}
The following is obvious:
\begin{equation*}
\dele{\sfe_1}\ge\deleo.
\end{equation*}
If these quantities are finite, one then has in the limit $\delta\to 0$:
\begin{equation}\label{unbiased-biased}
\dee{\sfe_1}\ge\deeo.
\end{equation}

The difference between $\dele{\sfe_1}$ and $\deleo$ disappears when the output 
distributions are concentrated at the locations of the input distributions, that is, around 
the intervals $\ixd$. This is to say that the difference is a measure of the overall magnitude
of the \emph{bias} $\bede{\sfe_1}$ inherent in $\sfe_1$ relative to $\sfe$:
\begin{equation*}
\bede{\sfe_1}:=\dele{\sfe_1}-\deleo\ge 0.
\end{equation*}
Rephrasing this as 
\begin{equation*}
\dele{\sfe_1}=\deleo+\bede{\sfe_1},
\end{equation*}
we see that the gross error is decomposed into the bias-free error and the magnitude of the
bias.  Note that one can take the limit of $\delta\to 0$:
\begin{equation*}
\bee{\sfe_1}:=\dee{\sfe_1}-\deeo.
\end{equation*}

As an immediate consequence of these definitions, we can say that $\sfe_1$ is an 
$\veps$-approximation to $\sfe$ if and only if the bias and random errors are finite
for all $\delta>0$. 

\begin{proposition}\label{prop:sm-bf-err}
Let $\sfq^\mu$, $\sfp^\nu$ be smeared position and momentum observables. Then
\begin{equation*}
\w_{\veps_1}(\sfq^\mu,\sfq)=W_{\veps_1}(\mu),
\quad \w_{\veps_2}(\sfp^\nu,\sfp)=W_{\veps_2}(\nu).
\end{equation*}
\end{proposition}

\emph{Proof.} It suffices to consider the case of position.
We show first that $W_{\veps_1}(\rho^{\sfq^\mu})\ge W_{\veps_1}(\mu)$. This is equivalent to the following:
$w\ge W_{\veps_1}(\rho^{\sfq^\mu})$ implies $w\ge W_{\veps_1}(\mu)$.

Thus, let $w$ be such that $\rho^{\sfq^\mu}(\iqw)\ge 1-\veps_1$ for some $q\in\R$. 
Assume $w<W_{\veps_1}(\mu)$; this means that for all $q'\in\R$ one has $\mu(\intv{q'}{w})<1-\veps_1$.
But then 
\begin{equation*}
\rho^{\sfq^\mu}(\intv{q'}{w})=\int \rho^{\sfq}(dx)\mu(x-\intv{q'}{w})<1-\veps_1,
\end{equation*}
which contradicts the premise.

Next we show that whenever $W_0(\rho^{\sfq})\le\delta$, then 
$W_{\veps_1}(\rho^{\sfq^\mu})\le W_{\veps_1}(\mu)+\delta$.
We are given that $\rho^{\sfq}(\intv{q_0}{\delta})=1$ for some $q_0\in\R$. Assume 
$w\ge W_{\veps_1}(\mu)$, that is, $\mu(\intv{q_1}{w})\ge1-{\veps_1}$ for some $q_1$.We have to show
that $w+\delta\ge W_{\veps_1}(\rho^{\sfq^\mu})$, that is, 
$\rho^{\sfq^\mu}(\intv{q_2}{w+\delta})\ge1-\veps_1$
for some $q_2\in\R$.

Let $q_2=q_0-q_1$. Then it follows that $q+\intv{q_2}{w+\delta}\supseteq \intv{q_0}{\delta}$ for
all $q\in\intv{q_1}{w}$. Then
\begin{eqnarray*}
\rho^{\sfq^\mu}(\intv{q_2}{w+\delta})&=\int\mu(dq)\rho^{\sfq}(q+\intv{q_2}{w+\delta})\\
&\ge
\int_{\iqw}\mu(dq)=\mu(\iqw)\ge1-\veps_1.
\end{eqnarray*}
This shows that $w\ge W_{\veps_1}(\mu)$ implies $w+\delta\ge W_{\veps_1}(\rho^{\sfq^\mu})$
whenever $W_0(\rho^\sfq)\le\delta$. Thus, under this assumption we let $w$ approach 
$W_{\veps_1}(\mu)$ to obtain $W_{\veps_1}(\mu)+\delta\ge W_{\veps_1}(\rho^{\sfq^\mu})$.

To summarise, we have shown:
$W_\veps(\mu)\le W_\veps(\rho^{\sfq^\mu})\le W_\veps(\mu)+\delta$, where the latter inequality
holds if $W_0(\rho^\sfq)\le\delta$. This entails that also
$W_\veps(\mu)\le \w_{\veps,\delta}(\sfq^\mu,\sfq)\le W_\veps(\mu)+\delta$. Now we can 
take the limit $\delta\to 0$ to obtain the result. \qed

\subsection{Measurement uncertainty relations for error bar widths}\label{sec:errorbar}

The following error relations for covariant approximations are special cases of the general result quoted below.
Their proofs are straightforward consequences of the considerations of this paper, hence we present them here
as separate statements.

\begin{proposition}
Let $\sfg^\bom$ be a covariant phase space observable. Then the bias-free error bar widths 
of the marginals relative to $\sfq$ and $\sfp$ obey the trade-off relation:
\begin{equation}\label{unbiased-ur}
\deq{\sfq^{\mu_\bom}}\,\dep{\sfp^{\nu_\bom}} \ge \w_{\veps_1}(\sfq^{\mu_\bom},\sfq)\, \w_{\veps_2}(\sfp^{\nu_\bom},\sfp)
 \ge 2\pi\hbar\,K(\veps_1,\veps_2),
\end{equation}
where $K(\veps_1,\veps_2)$ is given by Eq.~(\ref{eqn:ow-U}).
\end{proposition}

{\em Proof.}  The first inequality follows from (\ref{unbiased-biased}) and the second 
is a direct consequence  of Proposition \ref{prop:sm-bf-err} and Eq.~(\ref{eqn:gam-gam-cov-ur}) 

The corresponding inequality for general phase space observables was proven in \cite{BuPe07}.

\begin{theorem}\label{thm}
Let $\sfm$ be an approximate joint observable for $\sfq,\sfp$, in the sense that its marginals have finite error bar
widths as approximations of position and momentum, respectively. Then,
for $\veps_1,\veps_2\in(0,\frac 12)$,
the error bar widths of $\sfm_1$ and $\sfm_2$ satisfy the
uncertainty relation
\begin{equation*}
\deq{\sfm_1}\cdot\dep{\sfm_2}\ge
 2\pi\hbar\,K(\veps_1,\veps_2),
\end{equation*}
where $K(\veps_1,\veps_2)$ is given by Eq.~(\ref{eqn:ow-U}).
\end{theorem}

This result entails the following statement: an
approximate joint observable for $\sfq,\sfp$ cannot have one of
these sharp observables as its marginal. It is instructive to show this explicitly by
considering the case $\sfm_1=\sfq$.

\begin{proposition}\label{sharp-marg}
Let $\sfm$ be an observable on phase space whose first marginal
coincides with sharp position, $\sfm_1=\sfq$ (so that $\deq{\sfm_1}=0$).
Then the second marginal $\sfm_2$ cannot satisfy the condition of an
$\veps_2$-approximation to $\sfp$ for any $\veps_2\in(0,\frac 12)$, that is,
$\dep{\sfm_2}=\infty$. Hence $\sfm$ cannot be an
$(\veps_1,\veps_2)$-approximate joint observable to $\sfq,\sfp$ for
any $\veps_1,\veps_2\in(0,\frac 12)$.
\end{proposition}

\emph{Proof.} Let $\veps_2\in(0,\frac 12)$ be given and let $\delta>0$
and $w'>0$ be arbitrary. We have to show that there is an interval
$\ipd$ and a state $\rho$ localised in $\ipd$ so that $\tra{\rho
M_2(\ipw)}<1-\veps_2$.

As noted in Example \ref{ex3}, all positive operators (effects)
$\sfm_2(X)$ in the range of $\sfm_2$ commute with  $\qhat$  and are thus
functions of $\qhat$.
Thus we can write: $\sfm_2(X)=\int m(q,X)\,\sfq(dx)$. Consider the
sequence of intervals $J_{n;w'}$,  $n=0,1,2,\dots$. Since
$I=\sfm_2(\R)$, then $\sfm_2((-\infty,n-w'/2))\to1$ as $n\to\infty$
(ultraweakly), and it  follows that for every state $\rho$,
$\tra{\rho \sfm_2(J_{n;w'})}\le \tra{\rho \sfm_2([n-w'/2,\infty))}\to 0$,
hence:
\[
\tra{\rho \sfm_2(J_{n;w'})}=\int\rho^\sfq(dq)m(q,J_{n;w'})\to 0 \quad {\rm as}\ n\to\infty.
\]

Let $\rho_0$ be such that $\rho_0^\sfp(J_{0;\delta})=1$, that is, the
distribution $\rho_0^\sfp$ vanishes outside that interval. Then
$\rho_n:=W(0,n)\rho_0$ is localised in $J_{n;\delta}$, while the
position distribution is unchanged, $\rho_n^\sfq=\rho_0^\sfq$.

For the given $\veps_2\in(0,\frac 12)$, there is an $n\in\N$ such that for
the  fixed state $\rho_0$, $\tra{\rho_0 \sfm_2(J_{n;w'})}<1-\veps_2$.
Then, since $\rho_0^\sfq=\rho_n^\sfq$, we also have $\tra{\rho_n
\sfm_2(J_{n;w'})}<1-\veps_2$, whereas $\rho_n$ is localised in
$J_{n;\delta}$. \qed

\vspace{12pt}

This result reproduces, in particular, the well-known fact that
there is no observable on phase space whose marginals are sharp
position and sharp momentum.

\begin{example}
Example \ref{ex:noncov-approx} can be used to construct an observable  $M$ on phase space which is not 
covariant but is still an approximate joint observable for $\sfq,\sfp$. Let $\sfg^\bom$ be a
covariant phase space observable and define $\sfm:=\sfg^\bom\circ \gamma^{-1}$, where $\gamma(q,p):=
(\gamma_1(q),\gamma_2(p))$. We assume that $\gamma_1,\gamma_2$ are strictly increasing 
continuous functions such that $\gamma_1(q)-q$ and $\gamma_2(p)-p$ are bounded functions.
Then it follows  that the marginals $\sfm_1^\gamma=\sfg^\bom_1\circ\gamma_1^{-1}$ and
$\sfm_2^\gamma=\sfg^\bom_2\circ\gamma_2^{-1}$ have finite error bars with respect to $\sfq,\sfp$. If 
$\gamma$ is a nonlinear function then $\sfm$ will not be covariant.
\end{example}

It is straightforward to obtain a universal uncertainty relation for the bias-free errors for any
approximate joint observable $\sfm$ of $\sfq,\sfp$. The core of the proof is to show that finite
bias-free errors for the marginals entails the existence of a covariant observable $\sfg^\bom$ (obtained by
the operation of finite mean used in \cite{Werner04}) such that its marginals
are not greater than those of $\sfm$:
\begin{equation*}
\delqo{\sfm_1}\ge\delqo{\sfg^\bom_1},\quad \delpo{\sfm_2}\ge\delpo{\sfg^\bom_2}.
\end{equation*}
The proof of this is similar to that of Lemma 4 of \cite{BuPe07}  and will be omitted.
 Using inequality (\ref{unbiased-ur}), the bias-free errors is then seen to obey the trade-off relation for $\veps_1,\veps_2<\frac 12$: 
\begin{equation*}
\w_{\veps_1}(\sfm_1,\sfq)\, \w_{\veps_2}(\sfm_2,\sfp)\ge
{2\pi\hbar}\, K(\veps_1,\veps_2).
\end{equation*}

\subsection{Trade-off relations for resolution widths}

In the work \cite{BuPe07} we claimed the validity of an uncertainty relation for resolution widths; 
a proof was not given explicitly as it was considered to follow closely the steps of the proof of Theorem \ref{thm}. On revisiting this relation, we found that there is no obvious way of adapting that proof.
In fact, it may well be that it is only for sufficiently close joint approximations of $Q$ and $P$ 
that there have to be constraints on the resolution width similar to the error uncertainty relation.
Hence we rephrase the claim as a Problem.

\vspace{6pt}

\noindent{\bf Problem.} {\em 
Let observable $\sfm$ on $\brr$ be an approximate joint observable for $\sfq,\sfp$ in the sense of finite
error bar widths. State conditions on the quality of the approximation (other than the covariance of the joint observable) which entail that the resolution widths must obey the trade-off relation (for $\veps_1+\veps_2<1$):
\begin{equation*}
\gamma_{\veps_1}(\sfm_1)\,\gamma_{\veps_2}(\sfm_2)\ge 2\pi\hbar\,K(\veps_1,\veps_2)\quad (?)
\end{equation*}
}

\vspace{6pt}

\section{Error measures III: Noise-based error}


Classical statistical error analysis  is prominently based on the use of moments of probability distributions for
the quantification of measurement errors. Thus, a wide-spread approach
found in the literature of defining a measure of error is in terms of
a formal ``root mean square" deviation of an indicator variable $Z$ of the measuring 
apparatus from  the variable $A$ to be  measured approximately. Classically, $A$ and $Z$ are given as random variables, and quantum mechanically as selfadjoint operators: 
this  \emph{state-dependent noise-based error} is given as
the root mean square deviation, 
\begin{equation*}
\epno(A,\hM,\rho):=\langle (Z_{\rm out}-A_{\rm in})^2\rangle_{\rho\otimes\sigma}^{1/2}.
\end{equation*}
 Here $Z_{\rm out}$ denotes the output (pointer) observable at the end of the interaction
phase between object system and probe in the measurement $\mathcal{M}$, and $A_{\rm in}$ is the input object observable to be approximately measured; the object plus probe system is initially in the state $\rho\otimes\sigma$. The choice of name reflects the fact that the operator $Z_{\rm out}-A_{\rm in}$ is commonly called {\em noise operator}.

A detailed critique of this attempted quantum generalisation of the rms error is given in \cite{BLW2013a}; the main deficiency is that  this quantity
fails to be a faithful representation of the absence or magnitude of an approximation error. Therefore this measure has to be used with care; it is operationally significant only in some special circumstances; then it may be used to  provide estimates of measurement errors \cite{BLW2013a}.
Here we are concerned with the state-independent upper bound of the quantity $\epno$,
the \emph{(global) noise-based error} of a measurement $\mathcal{M}$ relative to $A$ as \cite{Appleby1998b}
\begin{equation*}
\epno(A,\mathcal{M}):=\sup_{\rho} \epno(A,\mathcal{M},\rho)
\end{equation*}
where the supremum is taken over all states $\rho$ for which the right hand side is well-defined.
We will say that the observable $\sfc$ defined by $\hM$ is a \emph{finite-noise approximation} to $A$ if
$\sfc$ has finite global noise-based error relative to $A$. 

The noise-based error $\epno(A,\mathcal{M},\rho)$ can be expressed in terms of the observable $\sfc$ actually measured by $\hM$
\cite{Ozawa04}:
\begin{equation}\label{OzaVeps1}
  \epno(A,\hM,\rho)^2= \tr{\rho(\sfc[x^2] -\sfc[x]^2)} + \tr{\rho (\sfc[x]-A)^2}.
\end{equation}

In order to apply this error measure in the case of joint approximate measurements, we note that if a measurement scheme $\hM$
defines an observable $\sfm$ on $\brr$, then its marginal observables $\sfm_1,\sfm_2$ can be taken as approximators for, say,
position $Q$ and momentum $P$, respectively. In this case the noise-based errors are defined via (\ref{OzaVeps1}) with $\sfc=\sfm_1$ for $\epno(Q,\hM,\rho)\equiv\epno(Q,\sfm_1,\rho)$ and with $\sfc=\sfm_2$ for $\epno(P,\hM,\rho)\equiv\epno(P,\sfm_2,\rho)$.
 We denote the global errors for the marginals of a general phase space observable by $\epno(Q,\sfm_1)$ and $\epno(P,\sfm_2)$, respectively. Then the following general result has been shown \cite{Appleby1998b}.

\begin{theorem}\label{thm:fin-st-err}
Let $\hM$ be a measurement realizing  an observable $\sfm$  on $\brr$. Then the
global noise-based errors obey the following trade-off relation.
\begin{equation*}
\epno(Q,\sfm_1)\,\epno(P,\sfm_2)\ge \frac \hbar 2.
\end{equation*}
The lower bound is realised for a covariant phase space observable $\sfg^\bom$ with $\bom$ being
the minimum uncertainty state operator with zero means of position and momentum.
\end{theorem}

\section{Connections}\label{connections}

We show that the concept of approximation based on
finite error bars generalises the notions of finite noise approximation
and metric approximations.

\begin{proposition}\label{prop:finite-dist-approx}
Any observable $\sfe_1$ on $\R$  that satisfies the condition 
$\Delta_\alpha(\sfe_1,\sfe)<\infty$ (for some $\alpha\in[1,\infty)$) for a sharp observable $\sfe$ on $\R$ is an 
approximation to $\sfe$ in the sense of finite error bars. In that case
the following inequality holds:
\begin{equation}\label{eqn:fin-dist-errbar}
\dee{\sfe_1}\le \frac{2}{\veps^{\frac1\alpha}}\,\Delta_\alpha(\sfe_1,\sfe)\,. 
\end{equation} 
\end{proposition}

\begin{proof} The proof is a straightforward adaptation of the proof for the case $\alpha=1$ given in \cite[Prop. 5]{BuPe07}.\\
 Using the definition of $\cD_\alpha(\rho^{\sfe_1},\rho^\sfe)$ and equation (\ref{DKanto}), we are given that
\begin{equation}\label{+}
\big|\tr{\rho \sfe_1[\Phi]}-\tr{\rho \sfe(\Psi)}\big|\le \Delta_\alpha(\sfe_1,\sfe)^\alpha=:c^\alpha,
\end{equation}
which holds for all $\rho\in S$ and all functions $\Psi,\Phi$ satisfying the constraint
\begin{equation}\label{++}
\biggl|\Phi(y)-\Psi(x)\biggr| \le |x-y|^\alpha,\quad x,y\in \R.
\end{equation}
Let $\veps\in(0,1)$ and $\delta>0$ be given. Put
$w=\delta+2n$, with $n\in\N,\ n^\alpha\ge c^\alpha/\veps$. 
Consider an interval $\iqd$ and a state $\rho$ with $\rho^\sfe(\iqd)=1$. Define
the functions $\Psi_n=\Phi_n\equiv h_n$ via
\[
h_n(x):=\left\{\begin{array}{ll} n^\alpha&{\rm if}\ \ |x-q|\le\delta/2;\\
\biggl[n+\delta/2-|x-q|\biggr]^\alpha&{\rm if}\ \ \delta/2<|x-q|\le \delta/2+n;\\
0&{\rm if}\ \ \delta/2+n<|x-q|.\\
\end{array}
\right.
\]
It is not hard to verify  that $\Psi_n,\Phi_n$ satisfy (\ref{++}). Condition $(\ref{+})$ for $\Psi_n=\Phi_n=h_n$ entails for
$g_n=h_n/n^\alpha$ that $\big|\tr{\rho {\sfe_1}[g_n]}-\tr{\rho \sfe[g_n]}\big|\le c^\alpha/n^\alpha$. We
then have 
$\chi_{\iqd}\le g_n\le \chi_{\iqw}$.

Now $\rho^\sfe(\iqd)=1$ implies $\tr{\rho \sfe(g_n)}=1$, and so, using the assumption
$n^\alpha\ge c^\alpha/\veps$,  we obtain
\[
\tr{\rho \sfe_1(\iqw)}\ge\tr{\rho \sfe_1(g_n)}\ge\tr{\rho \sfe(g_n)}-c^\alpha/n^\alpha\ge 1-\veps.
\]
To prove the inequality (\ref{eqn:fin-dist-errbar}), we note that on putting
$w=\delta+2c/(\veps^{1/\alpha})$, one still obtains $\tr{\rho \sfe_1(\iqw)}\ge
1-\veps$. This yields $\dele{\sfe_1}\le\delta+2\Delta_\alpha(\sfe_1,\sfe)/\veps^{1/\alpha}$, and on
letting $\delta$ approach 0, then (\ref{eqn:fin-dist-errbar}) follows.
\end{proof}

\begin{proposition}\label{prop:fin-st-err-err-bar}
Any observable $\sfe_1$ on $\R$ that satisfies the condition of finite global noise-based
error relative to a sharp observable with selfadjoint operator $A$ and associated spectral measure $\sfe$ (such that $A=\sfe[x]$), 
$\epno(A,\sfe_1)<\infty$, is an 
approximation to $\sfe$ in the sense of finite error bars. 
In that case, the  following inequality holds:
\begin{equation*}
\dee{\sfe_1}\le 2\epno(A,\sfe_1)\,\left(1+\sqrt{\frac{2}{\veps}}\right). 
\end{equation*}
\end{proposition}

\begin{proof}
We use the facts that $A^2=\sfe[x]^2=\sfe[x^2]$ and 
$\Delta(\sfe,\rho)=\Delta(A,\rho)$.

We begin by rewriting the definition of $\epno$ for general states $\rho$, denoted 
$\epno(\sfe_1,\sfe,\rho)$, and expressing the condition
of bounded errors: for all $\rho$ and $c:=\epno(\sfe_1,\sfe)<\infty$, 
\begin{eqnarray*}
\epno(A,\sfe_1,\rho)^2&=&\tra{\rho(\sfe_1[x]-A)^2}+\tra{\rho(\sfe_1[x^2]-\sfe_1[x]^2)}\\
&=&\tra{\rho(\sfe_1[x]-A)^2}+\Delta(\sfe_1,\rho)^2-\Delta(\sfe_1[x],\rho)^2\ \le\ c^2.
\end{eqnarray*}
(This follows readily from the corresponding condition stipulated for all vector states.)
The first term can be estimated as follows: using the inequality
\begin{eqnarray*}
|{\rm cov}_\rho(\sfe_1[x],A)|&=\tfrac 12 \left|\tra{\rho \sfe_1[x]\,A}+\tra{\rho A \sfe_1[x]}
-2\tra{\rho \sfe_1[x]}\tra{\rho A}\right|\\
&\le \Delta(\sfe_1[x],\rho)\Delta(A,\rho),
\end{eqnarray*}
we see that
\begin{eqnarray*}
&\tr[\rho(\sfe_1[x]-A)^2]=\Delta(\sfe_1[x]-A,\rho)^2+\left(\tra{\rho(\sfe_1[x]-A)}\right)^2\\
&\ =\Delta(\sfe_1[x],\rho)^2+\Delta(A,\rho)^2-2{\rm cov}_\rho(\sfe_1[x],A)
+\left(\tra{\rho(\sfe_1[x]-A)}\right)^2\\
&\ \ge\left(\Delta(\sfe_1[x],\rho)-\Delta(A,\rho)\right)^2+\left(\tra{\rho(\sfe_1[x]-A)}\right)^2.
\end{eqnarray*}
The boundedness of $\epno(\sfe_1,A,\rho)$ then gives:
\begin{eqnarray*}
\big(\Delta(\sfe[x],\rho)&-\Delta(A,\rho)\big)^2+\left(\tra{\rho \sfe_1[x]}-\tra{\rho A}\right)^2\\
&\ +\left(\Delta(\sfe_1,\rho)^2-\Delta(\sfe_1[x],\rho)^2\right)\le  \epno(\sfe_1,A,\rho)\le c^2.
\end{eqnarray*}
Each of the three bracketed terms is nonnegative and hence bounded above by $c^2$.
This implies: 
\[
\Delta(\sfe_1[x],\rho)^2-c^2\le\Delta(\sfe_1,\rho)^2\le\Delta(\sfe_1[x],\rho)^2+c^2,
\]
\[
\Delta(A,\rho)-c\le\Delta(\sfe_1[x],\rho)\le\Delta(A,\rho)+c,
\]
\begin{equation}\label{sharp3}
\tra{\rho A}-c\le\tra{\rho \sfe_1[x]}\le\tra{\rho A}+c;
\end{equation}
the first two inequalities taken together yield:
\begin{equation}\label{sharp4}
\Delta(\sfe_1,\rho)^2\le(\Delta(A,\rho)+c)^2+c^2.
\end{equation}
Now observe that the variance on the l.h.s. is the variance of the distribution $\prob:=\rho^{\sfe_1}$.
We use the following variant of Chebyshev's inequality, valid for any $w>0$:
\begin{eqnarray*}
\Delta(\prob)^2&=\int (x-\prob[x])^2\prob(dx)\\
&\ge\left\{\begin{array}{l}
\int_{\R\setminus\iqw}(x-\prob[x])^2\prob(dx)
\ge \left(\tfrac{w}2-|\prob[x]-q|\right)^2\big(1-\prob(\iqw)\big)\\ 
\hfill {\rm if\ } \prob[x]\in\iqw\,;\\
\int_{\iqw}(x-\prob[x])^2\prob(dx)
\ge \left(\tfrac{w}2-|\prob[x]-q|\right)^2\prob(\iqw) \\ 
\hfill {\rm if\ } \prob[x]\not\in\iqw\, .
\end{array}\right.
\end{eqnarray*}
We will only be using cases of large $w$  where $\prob[x]\in\iqw$ so that  we obtain:
\begin{equation}\label{sharp5}
\left(\tfrac{w}2-|\prob[x]-q|\right)^2\big(1-\prob(\iqw)\big)  \le\Delta(\prob)^2.
\end{equation}
Combining (\ref{sharp4})  and (\ref{sharp5})  yields:
\begin{equation}\label{sharp6}
\big(\tfrac w2-|\rho^{\sfe_1}[x]-q|\big)^2\,\big(1-\rho^{\sfe_1}(\iqw)\big)\le
(\Delta(A,\rho)+c)^2+c^2.
\end{equation}
We will only use this in the case of states $\rho$ for which  $\rho^\sfe(\iqd)=1$. In this case we
have $\Delta(A,\rho)\le\delta$ and $|\tra{\rho A}-q|\le\delta$, and using (\ref{sharp3})  we also obtain:
\[
|\rho^{\sfe_1}[x]-q| \le |\rho^{\sfe_1}[x]-\tra{\rho A}| + |\tra{\rho A}-q|\le c+\delta.
\]
We will also use only large (finite) $w$ so that we can assume
\[
\tfrac w2-|\rho^{\sfe_1}[x]-q| \ge \tfrac w2-(\delta+c)>0.
\]
Note that this entails, in particular, that $\rho^{\sfe_1}[x]\in\iqw$, so that the use of (\ref{sharp6}) 
is justified. Under these conditions (\ref{sharp6})  entails
\[
\big(\tfrac w2-(\delta+c)\big)^2\big(1-\rho^{\sfe_1}(\iqw)\big)  \le (\delta+c)^2+c^2.
\]
Now, for any $\veps$ one can choose $w$ large enough such that 
\[
\big({\tfrac w2}-(\delta+c)\big)^2=\frac{(\delta+c)^2+c^2}{\veps}
\]
Then $(\sharp 7)$ implies that $1-\rho^{\sfe}(\iqw)\le\veps$. Moreover, since
$\dele{E_1}\le w$, we also have
\[
\dele{\sfe_1}\le\frac 2{\sqrt{\veps}}\sqrt{(\delta+c)^2+c^2}+2(\delta+c),
\]
which in the limit $\delta\to 0$ yields
\[
\dee{\sfe_1}\le\left(1+\frac{\sqrt{2}}{\sqrt{\veps}}\right)\,2\epno(A,\sfe_1).
\]
\end{proof}
An interesting open question is whether 
finite global noise-based error also implies finite Wasserstein distances.

\section{Conclusion}

We have reviewed several measures of error and intrinsic unsharpness for 
measurements of position and momentum (or other observables supported on the
real line) and given a detailed investigation of their properties
and the relations between them.
We then have studied criteria for approximate (joint) measurements of
position and momentum, based on three different kinds of error measures:
Wasserstein $\alpha$-distances,  error bar width (with or without bias), and global noise-based error.
We have established two inequalities relating Wasserstein $\alpha$-distance and global noise-based error 
to error bar width, respectively, and have concluded that the criterion of finite error 
bars is satisfied whenever
the Wasserstein $\alpha$-distance or the  global noise error  is finite. Thus the criterion for approximate 
joint measurability of position and momentum in terms of finite error bars is the most general among the 
three. It is satisfied by \emph{all} covariant phase space observables whereas
for some of these observables the $\alpha$-distances or global noise  errors may be infinite.

For each of the three types of error measures we have reviewed a universal joint-measurement uncertainty 
relation. Put in geometric terms, these relations state that
the marginals $\sfm_1,\sfm_2$ of an observable $\sfm$ on phase space cannot both be
arbitrarily close to $\sfq,\sfp$, respectively.

We also considered the resolution width of an observable on $\R$, introduced in
\cite{CaHeTo06}, and posed the question under which assumptions on the quality of 
approximations for approximate joint measurements of position and momentum the  
resolution widths  of the marginals obey a Heisenberg-type uncertainty relation.

\section*{Acknowledgements}

It is a pleasure to thank Pekka Lahti for helpful comments on various manuscript versions of this work.

\section*{References}

\bibliographystyle{unsrt}


\end{document}